\documentclass[11pt]{article}

\usepackage{fullpage}
\usepackage{amsmath}
\usepackage{xspace}
\usepackage{amssymb}
\usepackage{epsfig}

\newtheorem{Thm}{Theorem}
\newtheorem{Lem}[Thm]{Lemma}

\newtheorem{Claim}{Claim}

\newtheorem{Def}{Definition}
\newtheorem{Fact}{Fact}

\newcommand{\qed}{\hfill{$\rule{6pt}{6pt}$}} 
\newenvironment{proof}{\noindent{\bf Proof}:}{\qed}

\mathchardef\mhyphen="2D

\newcommand{\sQ}{{\mathsf{Q}}}
\newcommand{\sR}{{\mathsf{R}}}

\newcommand{\suppress}[1]{}

\newcommand{\myand}{{\sf {AND}}}
\newcommand{\myor}{{\sf {OR}}}
\newcommand{\andor}{{\sf {AND \mhyphen OR}}}

\newcommand\DISJ{\mathsf{DISJ}}
\newcommand\NDISJ{\mathsf{NDISJ}}
\newcommand\IP{\mathsf{IP}}

\begin{document}
\title{Depth-Independent Lower bounds on the Communication Complexity of Read-Once Boolean Formulas}
\author{ Rahul Jain\thanks{Centre for Quantum Technologies and Department of Computer Science, National University of Singapore. Email: {\tt rahul@comp.nus.edu.sg}}  
\quad Hartmut Klauck\thanks{Centre for Quantum Technologies, National University of Singapore. Email: {\tt hklauck@gmail.com}}
\quad Shengyu Zhang\thanks{Department of Computer Science and Engineering, The Chinese University of Hong Kong. Email: {\tt syzhang@cse.cuhk.edu.hk}}}
\date{}
\maketitle

\abstract{We show lower bounds of $\Omega(\sqrt{n})$ and $\Omega(n^{1/4})$ on the randomized and quantum communication complexity, respectively, of all $n$-variable read-once Boolean formulas. Our results complement the recent lower bound of $\Omega(n/8^d)$ by Leonardos and Saks~\cite{LeonardosS09} and $\Omega(n/2^{\Omega(d\log d)})$ by Jayram, Kopparty and Raghavendra~\cite{JayramKR09} for randomized communication complexity of read-once Boolean formulas with depth $d$.

We obtain our result by ``embedding" either the Disjointness problem or its complement in any given read-once Boolean formula.
}
\section{Introduction}
A {\em read-once Boolean formula}  $f : \{0,1\}^n  \rightarrow \{0,1\}$ is a function which can be represented by a Boolean formula involving $\myand$ and $\myor$ such that each variable appears, possibly negated, at most once in the formula.  An {\em alternating $\andor$ tree} is a layered tree in which each internal node is labeled either $\myand$ or $\myor$ and the leaves are labeled by variables; each path from the root to the any leaf alternates between $\myand$ and $\myor$ labeled nodes. It is well known (see eg.~\cite{HeimanW91}) that given a read-once Boolean formula $f: \{0,1\}^n  \rightarrow \{0,1\}$ there exists a unique alternating $\andor$ tree, denoted $T_f$, with $n$ leaves labeled by input Boolean variables $z_1, \ldots, z_n$, such that the output at the root, when the tree is evaluated according to the labels of the internal nodes, is equal to $f(z_1 \ldots z_n)$. Given an alternating $\andor$ tree $T$, let $f_T$ denote the corresponding read-once Boolean formula evaluated by $T$. 

Let $x,y \in \{0,1\}^n$ and let $x \wedge y, x\vee y $ represent the {\em bit-wise} $\myand, \myor$ of the strings $x$ and $y$ respectively.  For $f : \{0,1\}^n  \rightarrow \{0,1\}$, let $f^\wedge : \{0,1\}^n \times \{0,1\}^n  \rightarrow \{0,1\} $ be given by $f^\wedge(x,y) = f(x \wedge y) $. Similarly let $f^\vee : \{0,1\}^n \times \{0,1\}^n  \rightarrow \{0,1\} $ be given by $f^\vee(x,y) = f(x \vee y) $. Recently Leonardos and Saks~\cite{LeonardosS09}, investigated the {\em two-party randomized communication complexity}, denoted $\sR(\cdot)$, of $f^\wedge, f^\vee $ and showed the following.  (Please refer to~\cite{KushilevitzN97} for familiarity with basic definitions in communication complexity.)

\begin{Thm}[\cite{LeonardosS09}]\label{fact:LS09}
Let $f : \{0,1\}^n  \rightarrow \{0,1\}$ be a read-once Boolean formula such that $T_f$ has depth $d$. Then
\[\max \{\sR(f^\wedge), \sR(f^\vee)\} \geq\Omega(n/8^d).\]
\end{Thm}
In the theorem, the depth of a tree is the number of edges on a longest path from the root to a leaf. Independently, Jayram, Kopparty and Raghavendra~\cite{JayramKR09} proved randomized lower bounds of $\Omega(n/2^{\Omega(d\log d)})$ for general read-once Boolean formulas and $\Omega(n/4^d)$ for a special class of ``balanced" formulas. 

It follows from results of Snir \cite{Snir85} and Saks and Wigderson \cite{SaksW86} (via a generic simulation of trees by communication protocols \cite{BCW98}) that for the read-once Boolean formula with their canonical tree being a {\em complete binary} alternating $\andor$ trees, the randomized communication complexity is $O(n^{0.753\ldots})$, the best known so far. However in this situation, the results of \cite{LeonardosS09, JayramKR09} do not provide any lower bound since $d = \log_2 n$ for the complete binary tree. We complement their result by giving universal lower bounds that do not depend on the depth. Below $\sQ(\cdot)$ represents the two-party {\em quantum communication complexity}.

\begin{Thm}\label{thm:AND-OR}
Let $f: \{0,1\}^n  \rightarrow \{0,1\}$ be a read-once Boolean formula. Then,
\[\max \{\sR(f^\wedge ), \sR(f^\vee )\} \geq\Omega(\sqrt n).\]
\[\max \{\sQ(f^\wedge ), \sQ(f^\vee ) \} \geq\Omega(n^{1/4}).\]
\end{Thm}

\vspace{0.1in}

\noindent {\bf Remark:} 
\begin{enumerate}
	\item Note that the maximum in Thoerem \ref{fact:LS09} and \ref{thm:AND-OR} is necessary since for example if  $f$ is the $\myand$ of the $n$ input bits then it is easily seen that $\sR(f^\wedge)$ is 1.
	\item This fact is easy to observe for balanced trees, as is also remarked in \cite{LeonardosS09}.
\end{enumerate}

\section{Proofs}
In this section we show the proof of Theorem~\ref{thm:AND-OR}. We start with the following definition.

\begin{Def}[Embedding]
We say that a function $g_1: \{0,1\}^r \times \{0,1\}^r  \rightarrow \{0,1\}$ can be {\em embedded} into a function  $g_2 :\{0,1\}^t \times \{0,1\}^t  \rightarrow \{0,1\}$, if there exist maps $h_a : \{0,1\}^r  \rightarrow \{0,1\}^t$ and $h_b : \{0,1\}^r  \rightarrow \{0,1\}^t$ such that $\forall x,y \in \{0,1\}^r, \; g_1(x,y) = g_2(h_a(x), h_b(y))$.
\end{Def}
It is easily seen that if $g_1$ can be embedded into $g_2$ then the communication complexity of $g_2$ is at least as large as that of $g_1$.

Let us define the {\em Disjointness} problem $\DISJ_n : \{0,1\}^n \times \{0,1\}^n  \rightarrow \{0,1\}$ as $\DISJ_n(x,y)=\bigwedge_{i=1,\ldots,n} (x_i\vee y_i)$ (where the usual negation of the variables is left out for notational simplicity). Similarly the {\em Non-Disjointness} problem $\NDISJ_n : \{0,1\}^n \times \{0,1\}^n  \rightarrow \{0,1\}$ is defined as $\NDISJ_n(x,y)=\bigvee_{i=1,\ldots,n} (x_i\wedge y_i)$. We shall also use the following well-known lower bounds. 
\begin{Fact}[\cite{KalyanasundaramS92, Razborov92}]\label{fact: random lb}
$\sR(\DISJ_n)=\Omega(n), \sR(\NDISJ_n)=\Omega(n)$.
\end{Fact}
\begin{Fact}[\cite{Razborov03}]\label{fact: quantum lb}
$\sQ(\DISJ_n)=\Omega(\sqrt n), \sQ(\NDISJ_n)=\Omega(\sqrt n)$.
\end{Fact} 

Recall that for the given read-once Boolean formula $f: \{0,1\}^n  \rightarrow \{0,1\}$ its the canonical tree is denoted $T_f$. We have the following lemma which we prove in Section \ref{sec: lemma}.
\begin{Lem}\label{lem:main}
\begin{enumerate}
\item Let $T_f$ have its last layer consisting only of $\myand$ gates. Let $m_0$ be the largest integer such that $\DISJ_{m_0}$ can be embedded into $f^\vee$ and $m_1$ be the largest integer such that $\NDISJ_{m_1}$ can be embedded into $f^\vee$. Then $m_0m_1\geq n$.
\item Let $T_f$ have its last layer consisting only of $\myor$ gates. Let $m_0$ be the largest integer such that $\DISJ_{m_0}$ can be embedded into $f^\wedge$ and $m_1$ be the largest integer such that $\NDISJ_{m_1}$ can be embedded into $f^\wedge$. Then $m_0m_1\geq n$.
\end{enumerate}
\end{Lem}
With this lemma, we can prove the lower bounds on $\max \{\sR(f^\wedge ), \sR(f^\vee )\}$ and $ \max \{\sQ(f^\wedge ), \sQ(f^\vee ) \}$ as follows. For an arbitrary read-once formula $f$ with $n$ variables, consider the sets of leaves 
\[
L_{odd} = \{\text{leaves in } T_f \text{ on odd levels}\}, \qquad  L_{even} = \{\text{leaves in } T_f \text{ on even levels}\}
\]
At least one of the two sets has size at least $n/2$; without loss of generality, let us assume that it is $L_{odd}$. Depending on whether the root is $\myand$ or $\myor$, this set consisting only of $\myand$ gates or $\myor$ gates, corresponding to case 1 or 2 in Lemma~\ref{lem:main}. Then by the lemma, either $\DISJ_{\sqrt{n/2}}$ or $\NDISJ_{\sqrt{n/2}}$ can be embedded in $f$ (by setting the leaves in $L_{even}$ to 0's). By Fact \ref{fact: random lb} and \ref{fact: quantum lb}, we get the lower bounds in Theorem ~\ref{thm:AND-OR}.

\subsection{Proof of Lemma ~\ref{lem:main}}\label{sec: lemma}
We shall prove the first statement; the second statement follows similarly. We first prove the following claim. 
\begin{Claim} \label{claim:main}
For an $n$-leaf $(n>2)$ alternating $\andor$ tree $T$ such that all its internal nodes just above the leaves have exactly two children (denoted the $x$-child and the $y$-child), let $s(T)$ denote the number of such nodes directly above the leaves. 
Let $m_0(T)$ be the largest integer such that $\DISJ_{m_0}$ can be embedded into $f_T$ and $m_1(T)$ be the largest integer such that $\NDISJ_{m_1}$ can be embedded into $f_T$. Then $m_0(T) m_1(T) \geq s(T)$.
\end{Claim}
\begin{proof}
The proof is by induction on depth $d$ of $T$. When $n>2$, the condition of the tree makes $d>1$, so the base case is $d=2$. 

\vspace{.1in} \noindent {\bf Base Case $d=2$:} In this case $T$ consists either of the root labeled $\myand$ with $s(T)$ (fan-in $2$) children labeled $\myor$s or it consists of the root labeled $\myor$ with $s(T)$ (fan-in $2$) children labeled $\myand$s. We consider the former case and the latter follows similarly. In the former case $f_T$ is clearly $\DISJ_{s(T)}$ and hence $m_0(T) = s(T)$. Also $m_1(T) \geq 1$ as follows. Let us choose the first two children $v_1, v_2 $ of the root. Further choose the $x$ child of $v_1$ and the $y$ child of $v_2$ which are kept free and the values of all other input variables are set to $0$. It is easily seen that the function (of input bits $x,y$) now evaluated is $\NDISJ_1$. Hence $m_0(T)m_1(T) \geq s(T)$.

\vspace{0.1in}

\noindent {\bf Induction Step $d > 2$:} Assume the root is labeled $\myand$ (the case when the root  is labeled $\myor$ follows similarly). Let the root have $r$ children $v_1, \ldots, v_r $ which are labeled $\myor$ and let the corresponding subtrees be $T_1, \ldots, T_r$ rooted at $v_1, \ldots, v_r $ respectively. Let without loss of generality the first $r'$ (with $0 \leq r' \leq r$) of these trees be of depth $1$ in which case the corresponding $s(\cdot) = 0$. It is easily seen that 
$$
	s(T) = r' + \left(\sum_{i=r'+1}^r s(T_i)\right)   \enspace.
$$
For $i>r'$, we have from the induction hypothesis that $m_1(T_i)m_0(T_i)\geq s(T_i)$. 

It is clear that  $m_0(T) \geq \sum_{i=1}^r m_0(T_i)$, since we can simply combine the Disjointness instances of the subtrees.
Also we have $m_1(T) \geq \max\{m_1(T_{r'+1}),\ldots, m_1(T_r),1\}$, because we can either take any one of the subtree instances (and set all other inputs to $0$), or at the very least can pick a pair of $x,y$ leaves (as in the base case above) and fix the remaining variables appropriately to realize a single $\myand$ gate which amounts to embedding $\NDISJ_1$. 
Now,
\begin{eqnarray*}
m_0(T)m_1(T) & \geq & \left(\sum_{i=1}^r m_0(T_i)\right) \cdot \left(\max\{m_1(T_1),\ldots, m_1(T_r),1\}\right) \\
& \geq & r' + \left(\sum_{i=r'+1}^r m_0(T_i)m_1(T_i)\right) \\
& \geq & r' + \left(\sum_{i=r'+1}^r s(T_i) \right) = s(T) \enspace .
\end{eqnarray*}
\end{proof}
\vspace{0.2in}

Now we prove Lemma~\ref{lem:main}: Let us view $f^\vee : \{0,1\}^{2n}  \rightarrow \{0,1\}$ as a read-once Boolean formula, with input $(x,y)$ of $f^\vee$ corresponding to the $x$- and $y$- children of the internal nodes just above the leaves. 
Note that in this case $T_{f^\vee}$ satisfies the conditions of the above claim and $s(T_{f^\vee}) = n$. Hence the proof of the first statement in Lemma~\ref{lem:main} finishes. 

\section{Concluding Remarks}
\begin{enumerate}
\item The randomized communication complexity varies between $\Theta(n)$ for the $\mathsf{Tribes}_n$ function (a read-once Boolean formula whose canonical tree has depth $2$) \cite{JKS03} and $O(n^{0.753\ldots})$ for functions corresponding to completely balanced $\andor$ trees (which have depth $\log n$).  It will probably be hard to prove a generic lower bound much larger than $\sqrt{n}$ for all read-once Boolean formulas $f : \{0,1\}^n \rightarrow \{0,1\}$, since the best known lower bound on the randomized query complexity of every read-once Boolean formula is $\Omega(n^{.51})$~\cite{HeimanW91} and communication complexity lower bounds immediately imply slightly weaker query complexity lower bounds (via the generic simulation of trees by communication protocols~\cite{BCW98}). 

\item Ambainis et al.~\cite{AmbainisCRSS07} show how to evaluate any alternating $\andor$ tree $T$ with $n$ leaves by a quantum query algorithm with slightly more than $\sqrt n$ queries; this also gives the same upper bound for the communication complexity of $\max\{\sQ(f_T^\wedge), \sQ(f_T^\vee)\}$. On the other hand, it is easily seen that the {\em parity} of $n$ bits can be computed by a formula of size $O(n^2)$ involving  $\myand, \myor$. Therefore it
is easy to show that the function {\em Inner Product modulo $2$} i.e. the function $\IP_m : \{0,1\}^m \times \{0,1\}^m  \rightarrow \{0,1\}$ given by $\IP_m(x,y)=\sum_{i=1}^m x_iy_i \mbox{ mod } 2$, with  $m = \sqrt n$ can be reduced to the evaluation of an alternating $\andor$ tree of size $O(n)$ and logarithmic depth. Since it is known that $\sQ(\IP_{\sqrt{n}}) = \Omega(\sqrt n)$~\cite{CleveDNT99}, we get an example of an alternating $\andor$ tree $T$ with $n$ leaves and $\log n$ depth such that $\sQ(f_T^\wedge) = \Omega(\sqrt{n})$. Since the same lower bound also holds for shallow trees such as $\myor$, hence $\Theta(\sqrt n)$ might turn out to be the correct bound on $\max\{\sQ(f_T^\wedge), \sQ(f_T^\vee)\}$ for all alternating $\andor$ trees $T$ with $n$ leaves regardless of the depth. 
\end{enumerate}

\vspace{0.1in}

\noindent {\bf Acknowledgments:}
Research of Rahul Jain and Hartmut Klauck is supported by the internal grants of the
Centre for Quantum Technologies (CQT), which is funded by the Singapore
Ministry of Education and the Singapore National Research Foundation. Research of Shengyu Zhang is partially supported by the Hong Kong grant RGC-419309, and partially by CQT for his research visit.
\bibliographystyle{alpha}
\bibliography{andor}

\newcommand{\etalchar}[1]{$^{#1}$}
\begin{thebibliography}{CvDNT99}

\bibitem[ACR{\etalchar{+}}07]{AmbainisCRSS07}
A.~Ambainis, A.~Childs, B.~Reichardt, R.~Spalek, and S.~Zhang.
\newblock Any and-or formula of size n can be evaluated in time $n^{1/2+o(1)}$
  on a quantum computer.
\newblock In {\em Proceedings of the 48th Annual IEEE Symposium on Foundations
  of Computer Science}, pages 363--372, 2007.

\bibitem[BCW98]{BCW98}
Harry Buhrman, Richard Cleve, and Avi Wigderson.
\newblock Quantum vs. classical communication and computation.
\newblock In {\em Proceedings of the Thirtieth Annual ACM Symposium on the
  Theory of Computing (STOC)}, pages 63--68, 1998.

\bibitem[CvDNT99]{CleveDNT99}
Richard Cleve, Wim van Dam, Michael Nielsen, and Alain Tapp.
\newblock Quantum entanglement and the communication complexity of the inner
  product function.
\newblock {\em Lecture Notes in Computer Science}, 1509:61--74, 1999.

\bibitem[HW91]{HeimanW91}
R.~Heiman and A.~Wigderson.
\newblock Randomized vs. deterministic decision tree complexity for read-once
  boolean functions.
\newblock In {\em Proceedings of the 6th Structures in Complexity Theory
  Conference}, pages 172--179, 1991.

\bibitem[JKR09]{JayramKR09}
T.S. Jayram, S.~Kopparty, and P.~Raghavendra.
\newblock On the communication complexity of read-once $ac^0$ formula.
\newblock In {\em Proceedings of the 24th Annual IEEE Conference on
  Computational Complexity}, pages 329--340, 2009.

\bibitem[JKS03]{JKS03}
Thathachar~S. Jayram, Ravi Kumar, and D.~Sivakumar.
\newblock Two applications of information complexity.
\newblock In {\em Proceedings of the thirty-fifth annual ACM symposium on
  Theory of computing (STOC)}, pages 673--682, 2003.

\bibitem[KN97]{KushilevitzN97}
Eyal Kushilevitz and Noam Nisan.
\newblock {\em Communication Complexity}.
\newblock Cambridge University Press, Cambridge, UK, 1997.

\bibitem[KS92]{KalyanasundaramS92}
B.~Kalyanasundaram and G.~Schnitger.
\newblock The probabilistic communication complexity of set intersection.
\newblock {\em SIAM Journal on Discrete Mathematics}, 5(4):545--557, 1992.
\newblock Earlier version in Structures'87.

\bibitem[LS09]{LeonardosS09}
N.~Leonardos and M.~Saks.
\newblock Lower bounds on the randomized communication complexity of read-once
  functions.
\newblock In {\em Proceedings of the 24th Annual IEEE Conference on
  Computational Complexity}, pages 341--350, 2009.

\bibitem[Raz92]{Razborov92}
A.~Razborov.
\newblock On the distributional complexity of disjointness.
\newblock {\em Theoretical Computer Science}, 106(2):385--390, 1992.

\bibitem[Raz03]{Razborov03}
A.~Razborov.
\newblock Quantum communication complexity of symmetric predicates.
\newblock {\em Izvestiya of the Russian Academy of Science, mathematics},
  67(1):159--176, 2003.
\newblock quant-ph/0204025.

\bibitem[Sni85]{Snir85}
M.~Snir.
\newblock Lower bounds for probabilistic linear decision trees.
\newblock {\em Theoretical Computer Science}, 38:69--82, 1985.

\bibitem[SW86]{SaksW86}
M.~Saks and A.~Wigderson.
\newblock Probabilistic boolean decision trees and the complexity of evaluating
  game trees.
\newblock In {\em Proceedings of the 27th Annual IEEE Symposium on Foundations
  of Computer Science}, pages 29--38, 1986.

\end{thebibliography}
\end{document}